\documentclass[10]{article}

\usepackage{amsmath}
\usepackage{amsthm}
\usepackage{amsfonts}
\usepackage{amssymb}

\usepackage{times}           

\usepackage{palatino}                     

\usepackage{tikz}
\usetikzlibrary{arrows,shapes}
\usetikzlibrary{decorations.markings}
\usetikzlibrary{calc}

\newtheorem{theorem}{Theorem}[section]

\newtheorem{proposition}[theorem]{Proposition}
\newtheorem{corollary}[theorem]{Corollary}
\theoremstyle{definition}
\newtheorem{definition}[theorem]{Definition}
\newtheorem{remark}[theorem]{Remark}
\newtheorem{example}[theorem]{Example}

\def\AA{{\mathcal A}}  \def\DD{{\mathcal D}}  \def\FF{{\mathcal F}}   
\def\OO{{\mathcal O}} \def\TT{{\mathcal T}}  \def\XX{{\mathcal X}}
\def\La{\Lambda}    
\def\Ld{\hbox{\it\char'44}^{}} 
\def\q{\partial}
\def\om{\;{\scriptstyle\circ}\;}
\def\w{\wedge}  
\def\wtilde{\widetilde}
\def\trr{\triangleright}
\def\ldb{[\![}  \def\rdb{]\!]}

\def\Ker{{\rm Ker}\,}
\def\Im{{\rm Im}\,}  
\def\Tr{{\rm Tr}\,}
\def\ad{{\rm ad}}
 
\def\id{{\rm id}}    
\def\spann{{\rm span}\,}
\def\rmEnd{\rm End\,}

\def\jjj{\,\vrule width6pt height.4pt depth0pt
        \vrule height6pt width.4pt depth0pt\,}    

\title{Canonical endomorphism field \\ on a Lie algebra}
\author{Jerzy Kocik\\
{\small Department of Mathematics, SIU, Carbondale, IL 62901, USA} \\
{\small E-mail: jkocik@siu.edu}}

\date{}
\begin{document}
\frenchspacing

\maketitle

\begin{abstract}
We show that every Lie algebra is equipped with a natural 
$(1,1)$-variant tensor field, the ``canonical endomorphism field", 
naturally determined by the Lie structure, 
and satisfying a certain Nijenhuis bracket condition.
This observation may be considered as complementary to 
the Kirillov-Kostant-Souriau theorem on symplectic 
geometry of coadjoint orbits. 
We show its relevance for classical mechanics, in particular for  Lax equations. 
We show that the space of Lax vector fields is closed under Lie bracket
and we introduce a new  bracket for vector fields on a Lie
algebra.  This bracket defines a new Lie structure on the space of vector fields.  


\end{abstract}


\paragraph{Notation.}
We shall distinguish between purely algebraic and 
differential products by using two types of brackets:

\noindent 
\begin{tabular}{cl}
 $\ldb \ \,, \ \rdb$            &---  Lie algebra product,\\[2pt]  
 $ \lbrack \ \,, \ \rbrack$  &--- Lie commutator of vector fields, Schouten bracket, Nijenhuis bracket.
\end{tabular}

\vspace{3pt}
\noindent
The summation convention over repeated indices is adopted throughout the paper. 

\section{Introduction} 

It is well-known 
that the underlying dual space $L^*$ of a  Lie 
algebra $L$ possesses --- as a manifold --- a canonical Poisson structure 
in terms of a smooth bi-vector field $\Omega\in\w^2TL^*$, which satisfies the Jacobi condition
$[\Omega,\Omega]=0$, and, when restricted to coadjoint orbits, is
nondegenerate and therefore invertible into a symplectic
structure  \cite{Tu1, Tu2, Li}.  
The existence of these symplectic sheets is the content
of the Kirillov-Kostant-Souriau Theorem \cite{CIMP, Ki, So}. 

In this paper we present an overlooked fact that the Lie algebra $L$ itself 
also possesses --- as a manifold --- a natural 
differential-geometric object, namely a $(1,1)$-type tensor field $\AA\in \TT^{(1,1)}L$
that we shall call the {\bf canonical endomorphism field} on $L$.
The principal geometric property of $\AA$ is that it is proportional to its own 
Nijenhuis derivative (Theorem \ref{th:2.1}). 

We discuss the relevance of this object for dynamical systems.
It turns out that what Hamilton equations are for the dual space $L^*$,
Lax equations are for $L$.
The principal property of $\AA$ assures that the space of 
``Lax vector fields" is closed under the Lie commutator
and, moreover, it allows one to introduce a new bracket of 
vector fields on $L$, which is the analog for Lax equations of the Poisson bracket on
Hamiltonian vector fields.

\section{The canonical endomorphism field on a Lie algebra}  

Customarily one defines a Lie algebra as a linear space $L$ 
with a product $L\times L\mapsto L$ denoted
$\ldb v,w\rdb$ (double bracket).
The product is bilinear, skew-symmetric (i),  
and satisfies the Jacobi identity (ii):
\begin{equation}	\label{eq:2.1}
\begin{aligned}
   {\rm{(i)}}  &\qquad \ldb v, w\rdb = - \ldb w, v\rdb  \\
   {\rm{(ii)}} &\qquad \ldb v, \ldb w, z\rdb\rdb
             + \ldb w, \ldb z, x\rdb\rdb
             + \ldb x, \ldb w, z\rdb\rdb=0  
\end{aligned}
\end{equation}
In a basis $\{e_i\}$, the commutator can be represented via ``structure
constants":
\begin{equation}	\label{eq:2.2}
                \ldb e_i, e_j\rdb = c^k_{ij}e_k
\end{equation}
Here we shall rather follow \cite{Ko} and define a Lie algebra as a pair $\{L,\,c\}$
where $c$ is a $(1,2)$-type tensor that in the above basis is 
\begin{equation}	\label{eq:22.1}
   c= {\scriptstyle\frac{1}{2}} c^k_{ij} \; \varepsilon^i\w \varepsilon^j\otimes e_k
\end{equation}
where $\{\varepsilon^i\}$ is the dual basis.  
The algebra product becomes a secondary, derived, concept:
$\ldb v,\, w \rdb = (v\w w)\jjj c = i_w  i_v c$. 
Similarly, the adjoint action of  $v\in L$
is defined simply as a (1,1)-tensor $\ad_v = v\jjj c$ in $L$.
Of course, from the structural point of view both definitions are equivalent,
$\{L,c\} \equiv \{L,\, \ldb\,\cdot\, ,\,\cdot\,\rdb\}$.

The point of the present paper is to look at the space $L$ as a flat manifold 
and consider various differential-geometric objects on it.
(We shall assume that $L$ is a real and finite dimensional.) 
The linear structure 
of this manifold allow one to prolong any tensor
$T$ in $L$ to the ("constant") tensor field $\wtilde T$ on the {\it manifold}  $L$. 
In particular, the manifold $L$ is equipped with 
a constant $(1,2)$-type tensor field $\lambda = \wtilde c$:
\begin{equation}	\label{eq:22.2}
   \lambda= {\scriptstyle\frac{1}{2}} c^k_{ij} \;  dx^i\w dx^j\otimes \q_k
\end{equation}
where $\{x^i\}$ are coordinates on $L$ associated with the basis $\{e_i\}$ 
and where we denote $\q_i \equiv \q / \q x^i$.
The manifold $L$ is also equipped with a natural vector field,
the Liouville vector field, which in a linear coordinate system is
\begin{equation}	\label{eq:22.222}
                     J= x^i\q_i
\end{equation}
Here is our basic observation:
\begin{theorem}  \label{th:2.1}
The manifold of the Lie algebra $L$ possesses a
natural field of endomorphisms (i.e., a (1,1)-variant tensor field)
$\AA\in\TT^{(1,1)}L$ defined by
\begin{equation}	\label{eq:22.33}
        \AA = J\jjj\lambda 
\end{equation}
Its Nijenhuis derivative $[\AA,\AA]$ is a vector-valued biform
\begin{equation}	\label{eq:22.44}
    \lbrack \AA,\AA] = -2 \lambda\jjj\AA
\end{equation}
Moreover $\AA$ acts on the adjoint orbits on $L$.
\end{theorem}

We shall call $\AA$ the {\bf canonical endomorphism field} on $L$.
In the coordinate description, $\AA$ and its Nijenhuis derivative are 
\begin{equation}	\label{eq:Abasic}
\begin{array}{rl}
               \AA   &=  \ x^i\, c^k_{ij} \; dx^j\otimes  \q_k \\[2pt]       
       [\AA,\AA] &=  \ - x^k\; c_{kp}^i c_{ab}^p  \;(dx^a\w  dx^b)\otimes\q_i
\end{array}
\end{equation}
The endomorphism field $\AA$ may be viewed as a family of local transformations that at point $x\in L$
can be represented by matrix $\AA_j^k(x) = x^i\, c^k_{ij}$.

Before we give its proof,  let us restate the theorem in more standard terms. 
The natural isomorphism of a tangent space at any $x\in L$  
with the space $L$ itself will be denoted by $\mu_x:T_xL\to L$.  
Then Theorem \ref{th:2.1}
states that every Lie algebra $L$ possesses, as a manifold, a unique 
natural tensor field $\AA\in\TT^{(1,1)}L$, which at point $x\in L$ 
is defined as an endomorphism taking a tangent vector $v\in T_xL$ to 
\begin{equation}	\label{eq:2.3}
        \AA_x(v) = (\mu_x^{-1}\om \ad_x\om \mu_x) (v)  
\end{equation}
or, in a somewhat sloppy notation,    $\AA(v)  = \ldb x,v\rdb$.
Its Nijenhuis derivative $[\AA,\AA]$ is a vector-valued biform
the evaluation of which equals  for any $v,w\in TL$:
\begin{equation}	\label{eq:2.4}
\begin{aligned}
    \lbrack \AA,\AA](v,w)  &= -2 \AA(\ldb v,w\rdb) \\
               &= -2 (\ldb \AA v,w\rdb + \ldb v,\AA w\rdb) 
\end{aligned}
\end{equation}
at point $x\in L$, the dependence of which was suppressed in the notation.

\begin{remark}
The canonical endomorphism field $\AA$ is defined for an arbitrary algebra and its 
differential-geometric properties, including the Nijenhuis bracket $[\AA,\AA]$, 
will reflect the type of this algebra.
In the present paper we restrict to Lie algebras where the Jacobi identity 
implies particularly pleasant consequences. 
\end{remark}

The above theorem may be viewed as a counterpart of the KKS theorem:
the essence of which is that the dual space $L^*$ is equipped with 
a bi-vector field $\Omega = x_k c^k_{ij} \q^i\w\q^j$  (in our language
$\Omega = J\jjj\lambda$).  Instead of the Nijenhuis bracket we have the Schouten bracket
$[\Omega,\,\Omega]_{Sch} = 0$. 
Thus $\Omega$ defines a Poisson structure, which, moreover,  restricts to the coadjoint orbits, 
on which its inverse $\omega$ defines a symplectic structure, $\omega = 0$.
Section \ref{s:duality} summarizes these parallels.

\section{Lie algebra in pictures}   

Tensor calculus gains much transparency when expressed in graphical language.

~\\
{\bf \large Basic Glyphs.}  
Here are the basic {\bf glyphs }corresponding to various tensors:
%
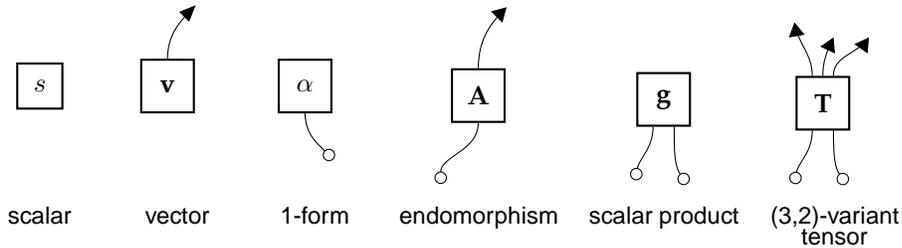
\begin{center}
\begin{tabular}{cccccc}
\phantom{\qquad\qquad}&
\phantom{\qquad\qquad}&
\phantom{\qquad\qquad}&
\phantom{\qquad\qquad}&
\phantom{\qquad\qquad}&
\phantom{\qquad\qquad}\\[-24pt]
\begin{tikzpicture}[baseline=-44pt, scale=1]
\node (a) at (0,0) [rectangle, draw, thick, minimum size=.6cm] {$s$};
\end{tikzpicture}
&
%
\begin{tikzpicture}[baseline=-44pt,scale=1]
\node (a) at (0,0) [rectangle, draw, thick, minimum size=.7cm] {\bf v};
\node  (b) at (.5,1.2) {};
\draw[-triangle 45] (a) to [out=90,in=225] (b);
\end{tikzpicture}
& 
%
\begin{tikzpicture}[baseline=-44pt,scale=1] 
\node (a) at (0,0) [rectangle, draw, thick, minimum size=.7cm] {$\mathbf \alpha$};
\node (b) at (.5,-1.1) {};
\draw[-o] (a) to [out=-90,in=135]  (b);
\end{tikzpicture}
&
%
\begin{tikzpicture}[scale=1] 
\node (a) at (0,0) [rectangle, draw, thick, minimum size=.7cm] {$\mathbf A$};
\node (bu) at (.5,1.3) {};
\node (bd) at (-.5,-1.3) {};
\draw[-triangle 60](a)to[out=90,in=-135](bu);
\draw[-o](a)to[out=-90,in=90](bd);
\end{tikzpicture}
&
%
\begin{tikzpicture}[scale=1] 
\node (g) at (0,0) [rectangle, draw, thick, minimum size=.7cm] {\ \bf g\ ~};
\node (b1) at (-.4,-1.2) {};
\node (b2) at (.4,-1.25) {};
\draw[-o] ($(g.south) - (.4em,0)$) to [out=-90,in=80] (b1);
\draw[-o] ($(g.south) + (.4em,0)$) to [out=-90,in=125](b2);
\end{tikzpicture}
&
%
\begin{tikzpicture}[scale=1] 
\node (g) at (0,0) [rectangle, draw, thick, minimum size=.7cm] {\ \bf T\ ~};
\node (u1) at (-.4,1.2) {};
\node (u2) at (.2,1) {};
\node (u3) at (.7,1) {};
\node (b1) at (-.4,-1.2) {};
\node (b2) at (.4,-1.2) {};
\draw[-triangle 60] ($(g.north) - (.4em,0)$)  to [out=90,in=-80] (u1);
\draw[-triangle 60] ($(g.north) + (0,0)$)       to [out=90,in=-115](u2);
\draw[-triangle 60] ($(g.north) + (.4em,0)$) to [out=90,in=-125](u3);

\draw[-o] ($(g.south) - (.4em,0)$) to [out=-90,in=80] (b1);
\draw[-o] ($(g.south) + (.4em,0)$) to [out=-90,in=125](b2);
\end{tikzpicture}
%
\\[-4pt]
\small\sf scalar &\small\sf vector &\small\sf 1-form &\small\sf endomorphism &\small\sf scalar product
&\small\sf (3,2)-variant 
\\[-4pt]
& & & & &\small\sf tensor
\end{tabular}
\end{center}
%

%
\noindent
where $s$ is a scalar, $\mathbf v$ is a vector, $\alpha$ is a covector, 
$A$ is an endomorphism, $g$ is a metric or biform.
The links with arrows  \ \tikz[baseline=-3pt]\draw[-triangle 45] (0,0) -- (7mm,0);
\ and links with circles \ \tikz[baseline=-3pt]\draw[-o] (0,0) -- (7mm,0); \ represent the contravariant 
and the covariant attributes of a tensor, respectively.
You may think of them as contravariant/covariant (upper/lower) indices 
in some basis description.
Scalars have none.

The ``in" and ``out" links may go any direction.  Turning and weaving in space does 
not have any meaning  (unlike in some other convections).  
For instance:

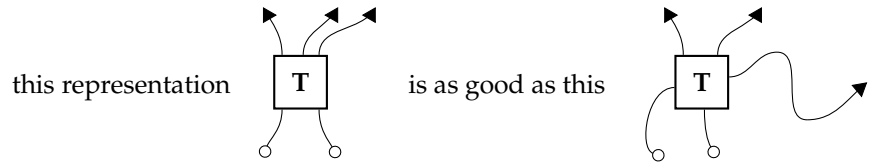
\begin{center}
this representation 
\hspace{.1in}
%
\begin{tikzpicture}[baseline=-4pt,scale=1] 
  \node (t) at (0,0) [rectangle, draw, thick, minimum size=.7cm] {\ \bf T ~};
  \coordinate (d1) at (-.5,-1);
  \coordinate (d2) at (.5,-1); 
  \coordinate (u1) at (-.5,1);
  \coordinate (u2) at (.5,1); 
  \coordinate (u3) at (1,1);
  \draw[-o] ($(t.south) - (.7em,0)$) to [out=-90,in=70] (d1);
  \draw[-o] ($(t.south) + (.7em,0)$)to[out=-90,in=120](d2);
  \draw[-triangle 60] ($(t.north) - (.7em,0)$) to [out=90,in=-60] (u1);
  \draw[-triangle 60] ($(t.north) + (.1em,0)$) to [out=90,in=-120] (u2);
  \draw[-triangle 60] ($(t.north) + (.7em,0)$) to [out=90,in=-120] (u3);
\end{tikzpicture}
\hspace{.1in}
             is as good as this
\hspace{.1in}
\begin{tikzpicture}[baseline=-4pt,scale=1] 
  \node (t) at (0,0) [rectangle, draw, thick, minimum size=.7cm] {\ \bf T\ ~};

  \coordinate (d1) at (-.5,-1);
  \coordinate (d2) at (.2,-1); 

  \coordinate (u1) at (-.5,1);
  \coordinate (u2) at (.8,1); 
  \coordinate (u31) at (1,.4);
  \coordinate (u32) at (1.4,-.5);
  \coordinate (u33) at (2.2,0);

  \draw[-o] ($(t.west) - (0, .2em)$)  to  [out=180,in=160] (d1);
  \draw[-o] ($(t.south) + (.1em,0)$) to [out=-90,in=120](d2);

  \draw[-triangle 60] ($(t.north) - (.7em,0)$) to [out=90,in=-60] (u1);
  \draw[-triangle 60] ($(t.north) + (.6em,0)$) to [out=90,in=-120] (u2);
  \draw[-] ($(t.east)  + (0,.2em)$) to [out=0, in=180] (u31);
  \draw[-] (u31) to [out=0,in=-180] (u32);
  \draw[-triangle 60] (u32) to [out=0,in=-140] (u33);
\end{tikzpicture}
\end{center}

The links may leave the box at any position, but the order of the point of departure is fixed: 
 the contravariant indices are ordered clockwise, while the covariant indices counterclockwise. 
Links may cross without any 
meaning implied.
\\

\noindent
Glyphs may be composed into {\bf pictograms} that represent 
terms resulting by manipulation with tensors.   
The tensor contractions are obtained by joining "ins" with "outs".
Here are some basic cases:
\\

\noindent
{\bf Evaluation:}  Here is the evaluation of a covector on a form:

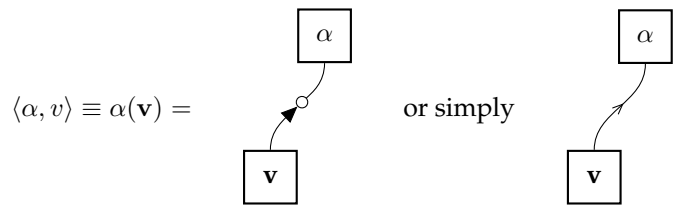
\begin{center}
$\langle\alpha,v\rangle \equiv \alpha(\mathbf v) = $
\hspace{.2in}
%
\begin{tikzpicture}[baseline=23pt, scale=.9, yscale=.95]
\node (a) at (0,0) [rectangle, draw, thick, minimum size=.7cm] {\bf v};
\coordinate (ab) at (.4,1.1);
\node (b) at (.8,2.2) [rectangle, draw, thick, minimum size=.7cm] {$\mathbf\alpha$};
\draw[-triangle 45](a)to[out=90,in=-135](ab);
\draw[-o](b)to[out=-90,in=45](ab);
\end{tikzpicture}
\hspace{.2in}
or simply
\hspace{.2in}
%
\begin{tikzpicture}[baseline=23pt,scale=.85, 
decoration={markings, mark=at position .55 with {\arrow{angle 45}};}]
  \node (v) at (0,0) [rectangle, draw, thick, minimum size=.7cm] {\bf v};
  \node (a) at (.8,2.2) [rectangle, draw, thick, minimum size=.7cm] {$\mathbf \alpha$};
  \draw[postaction={decorate}](v) to [out=90,in=-90](a);
\end{tikzpicture}
\end{center}

\newpage
\noindent
{\bf Scalar product:}  The scalar product of two vectors is a scalar $g(\mathbf v,\,\mathbf w) =$,
but if only one vector is contracted with $g$, then the result is a one-form: 

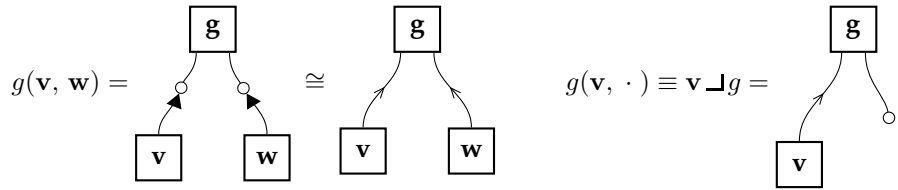
\begin{center}
$g(\mathbf v,\,\mathbf w) =$
%
\begin{tikzpicture}[baseline=-23pt,scale=.9, yscale=.95] 
\node (g) at (0,0) [rectangle, draw, thick, minimum size=.6cm] {\ \bf g\ ~};
\coordinate (g1) at (-.5,-1);
\coordinate (g2) at (.5,-1); 
\node (v) at (-.8,-2) [rectangle, draw, thick, minimum size=.6cm] {\bf v};
\node (w) at (.8,-2) [rectangle, draw, thick, minimum size=.6cm] {\bf w};
 \draw[-o]($(g.south) - (.7em,0)$) to [out=-90,in=70] (g1);
 \draw[-triangle 60] (v) to [out=90,in=-110] (g1);
 \draw[-o] ($(g.south) + (.7em,0)$)to[out=-90,in=120](g2);
 \draw[-triangle 60] (w) to [out=90,in=-60] (g2);
\end{tikzpicture}
\hspace{.005in}
                   $\cong$
\hspace{.005in}
%
\begin{tikzpicture}[baseline=-23pt, scale=.9, yscale=.9,
decoration={markings, mark=at position .55 with {\arrow{angle 45}};}]     
\node (g) at (0,0) [rectangle, draw, thick, minimum size=.6cm] {\ \bf g\ ~};
\node (v) at (-.8,-2) [rectangle, draw, thick, minimum size=.6cm] {\bf v};
\node (w) at (.8, -2) [rectangle, draw, thick, minimum size=.6cm] {\bf w};
\draw[postaction={decorate}](v) to [out=90,in=-90] ($(g.south) - (.7em,0)$) ;
\draw[postaction={decorate}](w) to [out=90,in=-90] ($(g.south) + (.7em,0)$);
\end{tikzpicture}
\hspace{.3in}
$g(\mathbf v,\;\cdot\;) \equiv \mathbf v \jjj g = $ 
%
\begin{tikzpicture}[baseline=-23pt,xscale=.9, yscale=.9,
decoration={markings, mark=at position .55 with {\arrow{angle 45}};}]   
\node (g) at (0,0) [rectangle, draw, thick, minimum size=.6cm] {\ \bf g\ ~};
\coordinate (g2) at (.54,-1.4); 
\node (v) at (-.8,-2) [rectangle, draw, thick, minimum size=.6cm] {\bf v};
 \draw[postaction={decorate}]  (v) to [out=90,in=-90] ($(g.south) -(.5em,0)$);
 \draw[-o]($(g.south) +(.5em,0)$) to [out=-90,in=100] (g2);
\end{tikzpicture}
\end{center}


~\\
{\bf Endomorphism} $A$ acting on a vector $v$ or covector $\alpha$ results in a vector or covector, respectively:
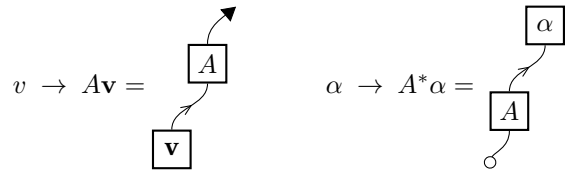
\begin{center}
$v\ \to\  A\mathbf v = $
%
\begin{tikzpicture}[baseline=-11pt,scale=.95, 
decoration={markings, mark=at position .55 with {\arrow{angle 45}};}]
  \node (A) at (0,0) [rectangle, draw, thick, minimum size=.5cm] {$A$};
  \node (v) at (-.5,-1.2) [rectangle, draw, thick, minimum size=.5cm] {\bf v};
  \coordinate (u) at (.4,.8);
  \draw[postaction={decorate}] (v) to [out=90,in=-90] (A);
  \draw[-triangle 60] (A) to [out=90,in=-140] (u);
\end{tikzpicture}
\hspace{.4in}
$\alpha \ \to \  A^*\alpha = \ $
\begin{tikzpicture}[baseline=7pt,scale=.95, 
decoration={markings, mark=at position .55 with {\arrow{angle 45}};}]
  \node (A) at (0,0) [rectangle, draw, thick, minimum size=.5cm] {$A$};
  \node (a) at (0.5,1.2) [rectangle, draw, thick, minimum size=.5cm] {$\alpha$};
  \coordinate (d) at (-.3,-.8);
  \draw[postaction={decorate}]  (A) to [out=90,in=-90] (a);
  \draw[-o] (A) to [out=-90,in=70] (d);
\end{tikzpicture}
\end{center}

~\\
{\bf Trace} may be represented by connecting ``in' with ``out" in a pictogram; 
If $A,B,C\in \rmEnd L$ are endomorphism of some linear space $L$, then we have:

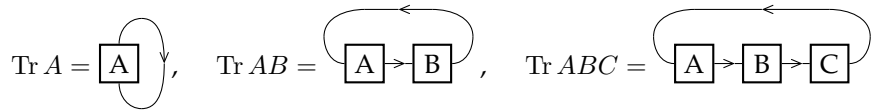
\begin{center}

$\Tr A=$\hspace{.01in}
\begin{tikzpicture}[baseline=-4pt,scale=.6] 
  \node (t1) at (0,0) [rectangle, draw, thick, minimum size=.5cm] {A}; 
  \coordinate (zap1) at (0.5,1);  
  \coordinate (zap2) at (1,0);  
  \coordinate (zap3) at (.5,-1);  
 \draw[-] (t1)  to  [out=90,in=180] (zap1); 
 \draw[-angle 45] (zap1)   to   [out=0,in=90] (zap2); 
 \draw[-] (zap2)   to  [out=-90,in=0] (zap3); 
 \draw[-] (zap3)   to  [out=180,in=-90] (t1); 
\end{tikzpicture}
,\hspace{.2in}$\Tr AB=$\hspace{.0in}
\begin{tikzpicture}[baseline=-4pt,scale=.9] 
  \node (t1) at (0,0) [rectangle, draw, thick, minimum size=.5cm] {A}; 
  \node (t2) at (1,0) [rectangle, draw, thick, minimum size=.5cm] {B}; 
  \coordinate (zip) at (.6,0);  
  \draw[-angle 45] ($(t1.east) + (0,0)$) to [out=0,in=180] (zip);
  \draw[-] ($(t2.west) - (0,0)$) to [out=180,in=0] (zip);
  \coordinate (zap1) at (1.6,.4);  
  \coordinate (zap2) at (1,.8);  
  \coordinate (zap3) at (.5,.8);  
  \coordinate (zap4) at (0,.8);  
  \coordinate (zap5) at (-.6,.4);  
 \draw[-] ($(t2.east) - (0,0)$)  to  [out=0,in=-90] (zap1); 
 \draw[-] (zap1)   to   [out=90,in=0] (zap2); 
 \draw[-angle 45] (zap2)   to  [out=180,in=0] (zap3); 
 \draw[-] (zap3)   to  [out=180,in=0] (zap4); 
 \draw[-] (zap4)   to  [out=180,in=90] (zap5); 
 \draw[-] (zap5)   to  [out=-90,in=180] (t1); 
\end{tikzpicture}
,\hspace{.2in}$\Tr ABC=$\hspace{.0in}
\begin{tikzpicture}[baseline=-4pt,scale=.9] 
  \node (t1) at (0,0) [rectangle, draw, thick, minimum size=.5cm] {A}; 
  \node (t2) at (1,0) [rectangle, draw, thick, minimum size=.5cm] {B}; 
  \node (t3) at (2,0) [rectangle, draw, thick, minimum size=.5cm] {C}; 
  \coordinate (zip1) at (.6,0);  
  \draw[-angle 45] ($(t1.east) + (0,0)$) to [out=0,in=180] (zip1);
  \draw[-] ($(t2.west) - (0,0)$) to [out=180,in=0] (zip1);
  \coordinate (zip2) at (1.6,0);  
  \draw[-angle 45] ($(t2.east) + (0,0)$) to [out=0,in=180] (zip2);
  \draw[-] ($(t3.west) - (0,0)$) to [out=180,in=0] (zip2);
  \coordinate (zap1) at (2.6,.4);  
  \coordinate (zap2) at (2,.8);  
  \coordinate (zap3) at (1,.8);  
  \coordinate (zap4) at (0,.8);  
  \coordinate (zap5) at (-.6,.4);  
 \draw[-] ($(t3.east) - (0,0)$)  to  [out=0,in=-90] (zap1); 
 \draw[-] (zap1)   to   [out=90,in=0] (zap2); 
 \draw[-angle 45] (zap2)   to  [out=180,in=0] (zap3); 
 \draw[-] (zap3)   to  [out=180,in=0] (zap4); 
 \draw[-] (zap4)   to  [out=180,in=90] (zap5); 
 \draw[-] (zap5)   to  [out=-90,in=180] (t1); 
\end{tikzpicture}
\end{center}

\noindent
The notable property of trace of a composition of endomorphisms, namely 
its invariance under cyclic permutation of the entries, 
$\Tr A_1\om...\om A_{k-1}\om A_k = \Tr A_2\om ...\om A_{k}\om A_1$, 
becomes in graphical language   
verifiable with a simplicity of a mantra on a {\it japa mala}. 

~\\
{\bf Lie algebra in pictures.} An {\bf algebra} is defined by a (1,2)-variant tensor $c$, as shown below on the left.
Also a product and adjoint representation is shown:

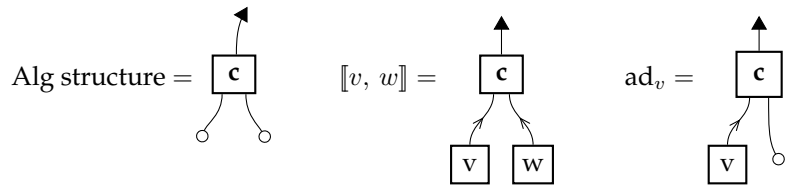
\begin{center}
Alg structure $=$
%
\begin{tikzpicture}[baseline=-4pt,scale=.9] 
  \node (t) at (0,0) [rectangle, draw, thick, minimum size=.5cm] {\ \bf c\ ~};
  \coordinate (d1) at (-.5,-1);
  \coordinate (d2) at (.5,-1); 
  \coordinate (u) at (.2,1);
    \draw[-o] ($(t.south) - (.5em,0)$) to [out=-90,in=70] (d1);
    \draw[-o] ($(t.south) + (.5em,0)$)to[out=-90,in=120](d2);
    \draw[-triangle 60] ($(t.north) + (.1em,0)$) to [out=90,in=-120] (u);
\end{tikzpicture}
\hspace{.3in} $\ldb v,\,w\rdb =$ \hspace{.0in}
%
\begin{tikzpicture}[baseline=-4pt,scale=.6,
decoration={markings, mark=at position .55 with {\arrow{angle 45}};}]
 
  \node (t) at (0,0) [rectangle, draw, thick, minimum size=.5cm] {\ \bf c\ ~};
  \node (v) at (-.7,-2) [rectangle, draw, thick, minimum size=.5cm] { v };
  \node (w) at (.7,-2) [rectangle, draw, thick, minimum size=.5cm] { w};
  \coordinate (u) at (0,1.3);
  \draw[postaction={decorate}]    (v) to [out=90,in=-90]  ($(t.south) - (.5em,0)$);
  \draw[postaction={decorate}]   (w) to [out=90,in=-90] ($(t.south) + (.5em,0)$);
  \draw[-triangle 60] ($(t.north) + (0,0)$) to [out=90,in=-90] (u);
\end{tikzpicture}
\hspace{.3in} $\hbox{ad}_v =$ \hspace{0in}
%
\begin{tikzpicture}[baseline=-4pt,scale=.6,
decoration={markings, mark=at position .55 with {\arrow{angle 45}};}]

  \node (t) at (0,0) [rectangle, draw, thick, minimum size=.6cm] {\ \bf c\ ~};
  \node (v) at (-.7,-2) [rectangle, draw, thick, minimum size=.5cm] { v };
  \coordinate (d2) at  (.5,-2); 
  \coordinate (u) at (0,1.3);

  \draw[postaction={decorate}] (v) to [out=90,in=-90]  ($(t.south) - (.6em,0)$);
  \draw[-o] ($(t.south) + (.6em,0)$)to[out=-90,in=120](d2);
  \draw[-triangle 60] ($(t.north) + (0,0)$) to [out=90,in=-90] (u);
\end{tikzpicture}
\end{center}
\noindent
If a single algebra is considered, the letter ``c" will be suppressed. 

~\\
In the case of a {\bf Lie algebra}, besides skew-symmetry we have the {\bf Jacobi identity},  
which may be written this way
\begin{equation}
\label{eq:jacobi}
\begin{tikzpicture}[baseline=-4pt,scale=1] 
  \node (t1) at (0,0) [rectangle, draw, thick, minimum size=.6cm] {\ \ ~}; 
  \coordinate (u1) at (1,0);  
  \coordinate (d11) at (-1,0);  
  \coordinate (d12) at (0,-.4); 
  \node (t2) at (0,-1) [rectangle, draw, thick, minimum size=.6cm] {\ \ ~};  
  \coordinate (u2) at (0,-.5);
  \node (d21) at (-.4,-2.1) {a};
  \node (d22) at (.4,-2.1)  { b}; 
  \draw[-triangle 60] ($(t1.east)$) to [out=0,in=180] (u1);  
  \draw[-o] ($(t1.west) - (0,0)$) to [out=180,in=0] (d11);
  \draw[-] ($(t1.south) + (0,0)$) to [out=-90,in=120](d12);
  \draw[-angle 45] ($(t2.north) + (0,0)$) to [out=90,in=-90] (d12);
  \draw[-o] ($(t2.south) - (.5em,0)$)  to    [out=-90,in=80] (d21); 
  \draw[-o] ($(t2.south) + (.5em,0)$) to [out=-90,in=100](d22);
\end{tikzpicture}
\hspace{.2in} = \hspace{.2in}
\begin{tikzpicture}[baseline=-4pt,scale=1] 
  \node (t1) at (0,0) [rectangle, draw, thick, minimum size=.6cm] {\ \ ~}; 
  \coordinate (u1) at (1,0);  
  \coordinate (d11) at (-1,0);  
  \node (d12) at (0,-1.2) {a}; 
  \node (t2) at (1,0) [rectangle, draw, thick, minimum size=.6cm] {\ \ ~}; 
  \coordinate (u2) at (2,0);  
  \coordinate (d21) at (.6,0);  
  \node (d22) at (1,-1.2) {b}; 
  \draw[-angle 45] ($(t1.east) + (0,0)$) to [out=0,in=180] (d21);
  \draw[-o] ($(t1.west) - (0,0)$)  to    [out=180,in=0] (d11); 
  \draw[-o] ($(t1.south) + (0,0)$) to [out=-90,in=90](d12);
  \draw[-triangle 60] ($(t2.east)$) to [out=0,in=180] (u2);  
  \draw[-] ($(t2.west) - (0,0)$) to [out=180,in=0] (d21);
  \draw[-o] ($(t2.south) + (0,0)$) to [out=-90,in=90](d22);
\end{tikzpicture}
\hspace{.1in} -  \hspace{.1in}
\begin{tikzpicture}[baseline=-4pt,scale=1] 
  \node (t1) at (0,0) [rectangle, draw, thick, minimum size=.6cm] {\ \ ~}; 
  \coordinate (u1) at (1,0);  
  \coordinate (d11) at (-1,0);  
  \node (d12) at (0,-1.2){b}; 
  \node (t2) at (1,0) [rectangle, draw, thick, minimum size=.6cm] {\ \ ~}; 
  \coordinate (u2) at (2,0);  
  \coordinate (d21) at (.6,0);  
  \node (d22) at (1,-1.2) {a}; 

  \draw[-angle 45] ($(t1.east) + (0,0)$) to [out=0,in=180] (d21);
  \draw[-o] ($(t1.west) - (0,0)$)  to    [out=180,in=0] (d11); 
  \draw[-o] ($(t1.south) + (0,0)$) to [out=-90,in=90](d12);
  \draw[-triangle 60] ($(t2.east)$) to [out=0,in=180] (u2);  
  \draw[-] ($(t2.west) - (0,0)$) to [out=180,in=0] (d21);
  \draw[-o] ($(t2.south) + (0,0)$) to [out=-90,in=90](d22);
\end{tikzpicture}
\end{equation}
The labels $a$ and $b$ are only to discern between different entries.
\\

\noindent
Perhaps the simplest derived object is a {\bf characteristic one-form} $\chi\in L^*$
the value of which on a vector $v\in L$ is $\chi(v) = - \Tr\hbox{ad}_v$.
Its pictograph is 
$$\chi \ = \ 
\begin{tikzpicture}[baseline=-4pt,scale=.7,
decoration={markings, mark=at position .7 with {\arrow{angle 45}};}]    

  \node (c) at (0,0) [rectangle, draw, thick, minimum size=.5cm] { ~ };
  \node (J) at (-.5, -1.5) { };
  \coordinate (zap1) at (.5,1);  
  \coordinate (zap2) at (1,.5);  
  \coordinate (zap3) at (1,-.5);  
  \coordinate (zap4) at (.5,-1);  
 \draw[-]                                      (c)  to  [out=90,in=180] (zap1)  to  [out=0,in=90] (zap2); ; 
 \draw[postaction={decorate}]  (zap2)   to   (zap3); 
 \draw[-]                                      (zap3)   to  [out=-90,in=0]  (zap4)   to  [out=180,in=-90]    ($(c.south) + (.6em,0)$);

  \draw[-o]   ($(c.south) - (.6em,0)$) to [out=-90,in=90] (J);
\end{tikzpicture}
$$
(This one-form vanishes for semisimple algebras.)

~\\
The {\bf Killing form} is defined as an inner product $K(v,\,w) =  \Tr \hbox{ad}_v\hbox{ad}_w$.
In the diagrammatic script it is easy to define the corresponding 2-covariant tensor $K$
\\

\begin{center}
$K$\hspace{.1in}$=$\hspace{.1in}
\begin{tikzpicture}[baseline=-4pt,scale=.9] 
  \node (t1) at (0,0) [rectangle, draw, thick, minimum size=.6cm] {\ \ ~}; 
  \node (t2) at (1,0) [rectangle, draw, thick, minimum size=.6cm] {\ \ ~}; 
 \node (d1) at (0,-1){}; 
 \node (d2) at (1,-1){}; 
\draw[-o] ($(t1.south) + (0,0)$) to [out=-90,in=90](d1);
\draw[-o] ($(t2.south) + (0,0)$) to [out=-90,in=90](d2);
  \coordinate (zip) at (.6,0);  
  \draw[->] (t1) to (zip);
  \draw[-] (zip) to (t2);
 
  \coordinate (zap1) at (1.6,.4);  
  \coordinate (zap2) at (1,.8);  
  \coordinate (zap3) at (.5,.8);  
  \coordinate (zap4) at (0,.8);  
  \coordinate (zap5) at (-.6,.4);  

 \draw[-] ($(t2.east) - (0,0)$)  to  [out=0,in=-90] (zap1); 
 \draw[-] (zap1)   to   [out=90,in=0] (zap2); 
 \draw[-angle 45] (zap2)   to  [out=180,in=0] (zap3); 
 \draw[-] (zap3)   to  [out=180,in=0] (zap4); 
 \draw[-] (zap4)   to  [out=180,in=90] (zap5); 
 \draw[-] (zap5)   to  [out=-90,in=180] (t1); 
\end{tikzpicture}
\end{center}

\noindent
Every Lie algebra possesses a skew-symmetric exterior {\bf Lie 3-form} $\omega$ that for any 
triple $v,w,z\in L$ takes value
$\omega(v,w,z) = \Tr \ad_{\ldb v,\,w\rdb}\ad_z$.
Using diagrammatic script we may ``draw" the form $\omega$ directly --- 
here it is, simplified with the use of Jacobi identity (\ref{eq:jacobi}): 
\\
\begin{center}
$\omega \ = \ $
\begin{tikzpicture}[baseline=-4pt,scale=.8] 
  \node (t1) at (0,0) [rectangle, draw, thick, minimum size=.5cm] {\ \ ~}; 
  \node (t2) at (1,0) [rectangle, draw, thick, minimum size=.5cm] {\ \ ~}; 
 \coordinate (d1) at (0,-.4){}; 
 \node (d2) at (1,-1){}; 
\draw[-] ($(t1.south) + (0,0)$) to [out=-90,in=90](d1);
\draw[-o] ($(t2.south) + (0,0)$) to [out=-90,in=90](d2);
  \coordinate (zip) at (.6,0);  
  \draw[-angle 45] ($(t1.east) + (0,0)$) to [out=0,in=180] (zip);
  \draw[-] ($(t2.west) - (0,0)$) to [out=180,in=0] (zip);
 
  \coordinate (zap1) at (1.6,.4);  
  \coordinate (zap2) at (1,.8);  
  \coordinate (zap3) at (.5,.8);  
  \coordinate (zap4) at (0,.8);  
  \coordinate (zap5) at (-.6,.4);  

 \draw[-] ($(t2.east) - (0,0)$)  to  [out=0,in=-90] (zap1); 
 \draw[-] (zap1)   to   [out=90,in=0] (zap2); 
 \draw[-angle 45] (zap2)   to  [out=180,in=0] (zap3); 
 \draw[-] (zap3)   to  [out=180,in=0] (zap4); 
 \draw[-] (zap4)   to  [out=180,in=90] (zap5); 
 \draw[-] (zap5)   to  [out=-90,in=180] (t1); 

  \node (t3) at (0,-1) [rectangle, draw, thick, minimum size=.5cm] {\ \ ~}; 
 \node (t3in1) at (-.3,-2){a}; 
 \node (t3in2) at (.3,-2){b};
\draw[-o] ($(t3.south) - (.4em,0)$) to [out=-90,in=90](t3in1);
\draw[-o] ($(t3.south) + (.4em,0)$) to [out=-90,in=90](t3in2);
\draw[-angle 45] (t3) to (d1);
\end{tikzpicture}
\hspace{.1in} $=$  \hspace{.1in}
\begin{tikzpicture}[baseline=-4pt,scale=.9] 
  \node (t1) at (0,0) [rectangle, draw, thick, minimum size=.5cm] {\ \ ~}; 
  \node (t2) at (1,0) [rectangle, draw, thick, minimum size=.5cm] {\ \ ~}; 
  \node (t3) at (2,0) [rectangle, draw, thick, minimum size=.5cm] {\ \ ~}; 
 \node (d1) at (0,-1){a}; 
 \node (d2) at (1,-1){b}; 
 \node (d3) at (2,-1){~}; 
\draw[-o] ($(t1.south) + (0,0)$) to [out=-90,in=90](d1);
\draw[-o] ($(t2.south) + (0,0)$) to [out=-90,in=90](d2);
\draw[-o] ($(t3.south) + (0,0)$) to [out=-90,in=90](d3);
  \coordinate (zip1) at (.6,0);  
  \draw[-angle 45] ($(t1.east) + (0,0)$) to [out=0,in=180] (zip1);
  \draw[-] ($(t2.west) - (0,0)$) to [out=180,in=0] (zip1);
  \coordinate (zip2) at (1.6,0);  
  \draw[-angle 45] ($(t2.east) + (0,0)$) to [out=0,in=180] (zip2);
  \draw[-] ($(t3.west) - (0,0)$) to [out=180,in=0] (zip2);
  \coordinate (zap1) at (2.6,.4);  
  \coordinate (zap2) at (2,.8);  
  \coordinate (zap3) at (1,.8);  
  \coordinate (zap4) at (0,.8);  
  \coordinate (zap5) at (-.6,.4);  
 \draw[-] ($(t3.east) - (0,0)$)  to  [out=0,in=-90] (zap1); 
 \draw[-] (zap1)   to   [out=90,in=0] (zap2); 
 \draw[-angle 45] (zap2)   to  [out=180,in=0] (zap3); 
 \draw[-] (zap3)   to  [out=180,in=0] (zap4); 
 \draw[-] (zap4)   to  [out=180,in=90] (zap5); 
 \draw[-] (zap5)   to  [out=-90,in=180] (t1); 
\end{tikzpicture}
\hspace{.1in} $-$  \hspace{.1in}
\begin{tikzpicture}[baseline=-4pt,scale=.9] 
  \node (t1) at (0,0) [rectangle, draw, thick, minimum size=.5cm] {\ \ ~}; 
  \node (t2) at (1,0) [rectangle, draw, thick, minimum size=.5cm] {\ \ ~}; 
  \node (t3) at (2,0) [rectangle, draw, thick, minimum size=.5cm] {\ \ ~}; 
 \node (d1) at (0,-1){b}; 
 \node (d2) at (1,-1){a}; 
 \node (d3) at (2,-1){~}; 
\draw[-o] ($(t1.south) + (0,0)$) to [out=-90,in=90](d1);
\draw[-o] ($(t2.south) + (0,0)$) to [out=-90,in=90](d2);
\draw[-o] ($(t3.south) + (0,0)$) to [out=-90,in=90](d3);
  \coordinate (zip1) at (.6,0);  
  \draw[-angle 45] ($(t1.east) + (0,0)$) to [out=0,in=180] (zip1);
  \draw[-] ($(t2.west) - (0,0)$) to [out=180,in=0] (zip1);
  \coordinate (zip2) at (1.6,0);  
  \draw[-angle 45] ($(t2.east) + (0,0)$) to [out=0,in=180] (zip2);
  \draw[-] ($(t3.west) - (0,0)$) to [out=180,in=0] (zip2);
  \coordinate (zap1) at (2.6,.4);  
  \coordinate (zap2) at (2,.8);  
  \coordinate (zap3) at (1,.8);  
  \coordinate (zap4) at (0,.8);  
  \coordinate (zap5) at (-.6,.4);  
 \draw[-] ($(t3.east) - (0,0)$)  to  [out=0,in=-90] (zap1); 
 \draw[-] (zap1)   to   [out=90,in=0] (zap2); 
 \draw[-angle 45] (zap2)   to  [out=180,in=0] (zap3); 
 \draw[-] (zap3)   to  [out=180,in=0] (zap4); 
 \draw[-] (zap4)   to  [out=180,in=90] (zap5); 
 \draw[-] (zap5)   to  [out=-90,in=180] (t1); 
\end{tikzpicture}
\end{center}

\noindent
where $a$ and $b$ are merely labels to distinguish the covariant entries.  
If we  use symbol the $\w$ or ``alt" inside a loop to denote the signed sum over all permutations 
of entries of a tensor (skewsymmetrization), then the Lie 3-covariant form is

\begin{center}
\begin{tikzpicture}[baseline=-4pt,scale=.9] 
  \node (t) at (0,0) [rectangle, draw, thick, minimum size=.6cm] {$\omega$};
  \coordinate (d1) at (-.5,-1);
  \coordinate (d2) at (0,-1); 
  \coordinate (d3) at (.5,-1); 

  \draw[-o] ($(t.south) - (.5em,0)$) to [out=-90,in=70] (d1);
  \draw[-o] ($(t.south) + (0,0)$)to[out=-90,in=90](d2);
  \draw[-o] ($(t.south) + (.5em,0)$)to[out=-90,in=120](d3);
\end{tikzpicture}
\hspace{.1in} $\equiv$  \hspace{.1in}
$\frac{1}{3} \ $
\begin{tikzpicture}[baseline=-4pt,scale=.9,
decoration={markings, mark=at position .7 with {\arrow{angle 45}};}]    
  \node (symb) at (.5,.5)  {alt}; 
  \node (t1) at (0,0) [rectangle, draw, thick, minimum size=.5cm] {\ \ ~}; 
  \node (t2) at (1,0) [rectangle, draw, thick, minimum size=.5cm] {\ \ ~}; 
  \node (t3) at (2,0) [rectangle, draw, thick, minimum size=.5cm] {\ \ ~}; 
 \node (d1) at (0,-1){}; 
 \node (d2) at (1,-1){}; 
 \node (d3) at (2,-1){~}; 
\draw[-o] ($(t1.south) + (0,0)$) to [out=-90,in=90](d1);
\draw[-o] ($(t2.south) + (0,0)$) to [out=-90,in=90](d2);
\draw[-o] ($(t3.south) + (0,0)$) to [out=-90,in=90](d3);
  \draw[-angle 45] (t1)  to (t2);
  \draw[-angle 45] (t2)  to (t3);
  \coordinate (zap1) at (2+.3,.8);  
  \coordinate (zap2) at (0-.3,.8);  
 \draw[-] ($(t3.east) - (0,0)$)  to  [out=0,in=0] (zap1); 
 \draw[->] (zap1)   to   (zap2); 
 \draw[-] (zap2)   to  [out=180,in=180] (t1); 
\end{tikzpicture}
\end{center}

\def\curvatureremark{
\noindent
Exploring some more pictograms one quickly arrives at this one:

\begin{center}
\quad $C = $
\begin{tikzpicture}[baseline=-14pt,scale=.9, 
decoration={markings, mark=at position .7 with {\arrow{angle 45}};}]    
  \node (c1) at (1,0) [rectangle, draw, thick, minimum size=.4cm] {\ ~}; 
  \node (c2) at (2,0) [rectangle, draw, thick, minimum size=.4cm] {\ ~}; 
  \node (cc1) at (1,-1) [rectangle, draw, thick, minimum size=.4cm] {\ ~}; 
  \node (cc2) at (2,-1) [rectangle, draw, thick, minimum size=.4cm] {\ ~}; 
  \coordinate (ccd11) at (1-.2, -1.8)   ;
  \coordinate (ccd12) at (1+.2, -1.8); 
  \coordinate (ccd21) at (2-.2, -1.8) ;
  \coordinate (ccd22) at (2+.2, -1.8) ; 
  \draw[postaction={decorate}] (c1) to (c2);
  \coordinate (zap1) at (2+.3,.8);  
  \coordinate (zap2) at (1-.3,.8);  
 \draw[-]                                      (c2)  to  [out=0,in=0] (zap1); 
 \draw[postaction={decorate}]  (zap1)   to   (zap2); 
 \draw[-]                                      (zap2)   to  [out=180,in=180] (c1); 
  \draw[postaction={decorate}] (c1) to (cc1);
  \draw[postaction={decorate}] (c2) to (cc2);
  \draw[-o] ($(cc1.south) - (.3em,0)$)   to [out=-90,in=90] (ccd11);
  \draw[-o] ($(cc1.south) + (.3em,0)$)  to [out=-90,in=90] (ccd12) ;
  \draw[-o] ($(cc2.south) - (.3em,0)$)  to[out=-90,in=90] (ccd21);
  \draw[-o] ($(cc2.south) + (.3em,0)$) to[out=-90,in=90] (ccd22) ;
\end{tikzpicture}
\end{center}
Does it have any significance?  One readily observes that it has symmetry property of a curvature tensor:
it is skewsymmetric with respect to the the order change of the first and of the second couple
of entries, and it is symmetric with respect to the change of the couple.
Quite interestingly, it does actually represent a curvature tensor of the 
Cartan connection on a corresponding Lie group (up to a scalar, depending on 
particular representation).
\\ 

}


\section{Differential geometry on a Lie algebra}  

Let us now look at the differential geometry of Lie algebra viewed as a manifold.
In the diagrammatic language the objects of Theorem \ref{th:2.1} are 
\begin{center}
{\sf Definition:}\quad  $\mathcal{A} =$ \hspace{0in}
%
\begin{tikzpicture}[baseline=-4pt,scale=.6] 
  \node (t) at (0,0) [rectangle, draw, thick, minimum size=.5cm] { ~ };
  \node (J) at (-.7,-2) [rectangle, draw, thick, minimum size=.5cm] { J };
  \coordinate (d1) at  (-.4,-1);
  \coordinate (d2) at  (.5,-2); 
  \coordinate (u) at (0,1.3);
  \draw[-] ($(t.south) - (.6em,0)$) to [out=-90,in=60] (d1);
  \draw[-o] ($(t.south) + (.6em,0)$)to[out=-90,in=90](d2);
  \draw[-angle 45] (v) to [out=90,in=-120] (d1);
  \draw[-triangle 60] ($(t.north) + (0,0)$) to [out=90,in=-90] (u);
\end{tikzpicture}
\qquad
{\sf Theorem:}\quad $\frac{1}{2}[\mathcal{A},\mathcal{A}] =$ \hspace{0in}
%
\begin{tikzpicture}[baseline=-4pt,scale=.6] 
  \node (t1) at (0,0) [rectangle, draw, thick, minimum size=.5cm] { ~ };
  \node (t2) at (.7, -2) [rectangle, draw, thick, minimum size=.5cm] { ~ };
  \node (J) at (-.7,-2) [rectangle, draw, thick, minimum size=.5cm] { J };
  \coordinate (d1) at  (-.4,-1);
  \coordinate (d2) at  (.4,-1);
  \coordinate (dd1) at  (0.2,-3.3);
  \coordinate (dd2) at  (0.9,-3.3);

  \coordinate (u) at (0,1.3);
  \draw[-triangle 60] ($(t1.north) + (0,0)$) to [out=90,in=-90] (u);
  \draw[-] ($(t1.south) - (.6em,0)$) to [out=-90,in=60] (d1);
  \draw[-] ($(t1.south) + (.6em,0)$)to[out=-90,in=120](d2);
  \draw[-angle 45] (J) to [out=90,in=-120] (d1);
  \draw[-angle 45] (t2) to [out=90,in=-60] (d2);
  \draw[-o] ($(t2.south) - (.6em,0)$) to [out=-90,in=60] (dd1);
  \draw[-o] ($(t2.south) + (.6em,0)$)to[out=-90,in=90](dd2);
\end{tikzpicture}
\end{center}

\noindent
Since the contraction with $J$ is introduces dependence on poisition 
(coordinates $x$), we shall use rather notation that will be 
easier perceptually.  Thus, for instance:
\begin{center}
$\mathcal{A} =$ \hspace{0in}
%
\begin{tikzpicture}[baseline=-4pt,scale=.6,, yscale=.9,
decoration={markings, mark=at position .7 with {\arrow{angle 45}};}]    
  \node (t) at (0,0) [rectangle, draw, thick, minimum size=.5cm] { ~ };
  \node (J) at (-.7,-2) [rectangle, draw, thick, minimum size=.5cm] { J };
  \coordinate (d1) at  (-.4,-1);
  \coordinate (d2) at  (.5,-2); 
  \coordinate (u) at (0,1.3);
  \draw[postaction={decorate}] (J) to [out=90,in=-90] ($(t.south) - (.6em,0)$);
  \draw[-o] ($(t.south) + (.6em,0)$)to[out=-90,in=90](d2);
  \draw[-triangle 60] ($(t.north) + (0,0)$) to [out=90,in=-90] (u);
\end{tikzpicture}
\qquad$\equiv$\qquad
%
\begin{tikzpicture}[baseline=-4pt,scale=.6, yscale=.9] 
  \node (t) at (0,0) [rectangle, draw, thick, minimum size=.5cm] { ~ };
  \node (J) at (-.7,-2)  { $x$ };
  \coordinate (d2) at  (.5,-2); 
  \coordinate (u) at (0,1.3);
  \draw[-] ($(t.south) - (.6em,0)$) to [out=-90,in=60] (d1);
  \draw[-o] ($(t.south) + (.6em,0)$)to[out=-90,in=90](d2);
  \draw[-angle 45] (v) to [out=90,in=-120] (d1);
  \draw[-triangle 60] ($(t.north) + (0,0)$) to [out=90,in=-90] (u);
\end{tikzpicture}
\end{center}

%
Every element (vector)  $v\in L$ 
defines a ``constant" vector field  $\wtilde v\in\XX L$  on manifold $L$
obtained by parallel transport; 
in coordinates,, if $v = v^i e_i$ then $\wtilde v = v^i \q_i$.
The canonical endomorphism field $\AA$ on manifold $L$ applied to such fields 
defines a representation of Lie algebra $L$ in terms of vector fields on $L$, namely with every algebra
element $v\in L$, we associate a vector field
\begin{equation}	\label{eq:3.4}
\begin{aligned}
  X_v &= \AA \wtilde v \\
      &= x^i v^j c^k_{ij} \q_k \quad = 
\end{aligned}
\qquad\qquad\qquad
\begin{tikzpicture}[baseline=-4pt,scale=.6, yscale=.9,
decoration={markings, mark=at position .7 with {\arrow{angle 45}};}]    
  \node (t) at (0,0) [rectangle, draw, thick, minimum size=.5cm] { ~ };
  \node (J) at (-.7,-2)  {$x$ };
  \node (v) at  (.5,-2) [rectangle, draw, thick, minimum size=.5cm] {v}; 
  \coordinate (u) at (0,1.3);
  \draw[postaction={decorate}] (J) to [out=90,in=-90] ($(t.south) - (.6em,0)$) ;
  \draw[postaction={decorate}]  (v) to [out=90,in=-90]  ($(t.south) + (.6em,0)$);
  \draw[-triangle 60] (t) to [out=90,in=-90] (u);
\end{tikzpicture}
\end{equation}

\begin{proposition}
\label{prop:homo}
The map $v\to X_v$ defines a homomorphism $\{L, \,\ldb\,\cdot\,,\,\cdot\,\rdb\} \to \{\XX L, [\,\cdot\,,\,\cdot\,]\}$ 
(the {\it infinitesimal representation} of $L$ in terms of $\XX L$):
\begin{equation}	\label{eq:3.5}
  [X_v,X_w] = X_{\ldb v,w\rdb}
\end{equation}
If the center of $L$ is trivial, the map presents a monomorphism. 
\end{proposition}
\begin{proof} 
The proposition readily  follows from the Jacobi identity.
\end{proof}

\begin{corollary} 	\label{cor:formulae}
The following are convenient formulae
\begin{align*}
   (i) &\qquad  [X_v,\,\wtilde w] = \wtilde{\;\ldb v,w \rdb\;} \\
   (ii)&\qquad \ldb\wtilde v,\wtilde w \rdb = \wtilde{\;\ldb v,w \rdb\;} \\
   (iii) &\qquad  [\wtilde v,\wtilde w ] = 0 
\end{align*}
\end{corollary}

The image of $\AA$ spans at every point a
subspace of the tangent space of $L$, defining in this way a distribution 
\begin{equation}	
\label{eq:3.6}
\DD =  \Im \AA  = \spann\{ X_v \mid v\in L\}
\end{equation}
The integral
manifolds of this distribution coincide with the adjoint orbits
determined by the action of a Lie group on Lie algebra.
Note however that we may define ``adjoint orbits" without reference to the Lie
group simply as the integral manifolds $\OO$ of $\DD$, satisfying $T\OO = \DD$.
\\

Now we prove the theorem.

\begin{proof}[Proof of Theorem 2.1:]
Recall that the Nijenhuis bracket $[K,K]$ of a vector-valued one-form
(endomorphism field) $K$ with itself is a vector-valued bi-form that,
evaluated on two fields $X$ and $Y$, takes the value according to 
\begin{equation}	\label{eq:3.7}
        {1\over 2}[K,K](X,Y) = [KX,KY] -K[KX,Y] -K[X,KY]+K^2[X,Y]
\end{equation}	
(see, e.g., \cite{Mi}).
Evaluating (half of) the Nijenhuis bracket
$[\AA,\,\AA]$ on two constant vector fields $\wtilde v$ and $\wtilde w$
and using formulae of Proposition \ref{prop:homo} and \ref{cor:formulae}, one gets
\begin{align*}
   {1\over 2} [\AA,\AA](\wtilde v,\wtilde w)
            &=[X_v,X_w] - \AA([X_v,\wtilde w]) - \AA([\wtilde v,X_ w]) +0\\
            &=X_{\ldb v,w\rdb} - X_{\ldb v,w \rdb} - X_{\ldb v,w \rdb} \\
            &=-X_{\ldb v,w\rdb}  
\end{align*}
In particular, substitution $X=\q_a$ and $Y=\q_b$ leads to the
coordinate formula \eqref{eq:Abasic}.
Now, let us show that $\AA$ can be restricted to orbits, i.e.,
$
        \AA (T_x\OO)\subset T_x\OO
$
for each point $x\in\OO$. First, rewrite \eqref{eq:2.3} for $X\in T_xL$:
\[
        \AA(X) = \mu_x^{-1}\left(\ldb x,(\mu_x(X)\rdb\right)
\]
Vector of the vector field $X_v$ at point $x\in L$ can be expressed as
\[
     X_v(x) = \AA(v) = \mu_x^{-1}\left(\ldb x, v\rdb\right)
\]
Thus
\[
     \AA(X_v) = \mu_x^{-1}\left(\ldb x,\mu_x (X_v)\rdb\right)
                    = \mu_x^{-1}\left(\ldb x,\ldb x,v \rdb\rdb\right)
                   = X_{\ldb x,v\rdb} \in T_x\OO  
\]
which was to be proven.
\end{proof}

\begin{example}	\label{exm:3.2}
Consider the 2-dimensional solvable algebra defined by
$[e_1,e_2]=e_2$.  Then
\begin{align*}
         \AA &= x_1 \; dx_2 \otimes \q_2 \\
        [\AA,\AA] &= -2 x_1 \; (dx_1\w dx_2) \otimes \q_2 
\end{align*}
The adjoint orbits are lines parallel to $e_2$ and the canonical endomorphism --- when restricted to 
any of them --- becomes a dilation.
\end{example}

\begin{example}	\label{exm:3.3}
The Lie algebra of 3-dimensional rotations, $so_3$, is
defined by relations
$\rdb e_i,e_j \ldb=\varepsilon_{ijk}\,e_k$.  Thus
\begin{align*}
         \AA &= x_1 \; (dx^2 \otimes \q_3 -  dx^3 \otimes \q_2)
                           + \hbox{(cyclic terms)}    \\
        [\AA,\AA] &= dx^1\w dx^2 \otimes (x^1\q_2-x^2\q_1)
                           + \hbox{(cyclic terms)}  
\end{align*}
The orbits are spheres defined by the Killing form.  On the unit sphere, tensor $\AA$
forms an almost complex structure, $\AA\om\AA=-\id$.
\end{example}

\begin{remark}
Although the Nijenhuis bracket (\ref{eq:22.44}) vanishes for two-step nilpotent algebras, 
(including Hei\-sen\-berg-type algebras \cite{Ka, FKS}), in general it does not, 
and therefore endomorphism field $\AA$ is in general not integrable.
Note that for vector fields of infinitesimal representation, the bi-form \eqref{eq:22.44}
takes at any point $x$ a vector-value
\begin{equation}	\label{eq:4.2}
    [\AA,\AA](X_v,X_w)
       =\mu_x^{-1}\om \ldb x,\ldb\ldb x,v\rdb,\ldb x,v \rdb\rdb\rdb
\end{equation}
Thus $\AA$ restricted to an orbit $\OO\subset L$ is (locally)
integrable if $\ldb x,\ldb \ldb x,v \rdb,\ldb x,v\rdb\rdb\rdb=0$ 
for every $x\in\OO$ and every $v,w\in L$.  
This is true for $so(n)$, $n\leq 4$ and for nilpotent
algebras of the upper-triangular $n\times n$ matrices, $n\leq5$.
\end{remark}

\section{Other basic properties of the endomorphism field}  

The fundamental property of the canonical endomorphism field (Theorem \ref{th:2.1}) is 
\[
              [\AA,\AA]=-2\lambda\jjj\AA
\]
Other basic properties of the geometry of a Lie algebra are summarized below:

\begin{corollary}	\label{cor: 4.2} 
The endomorphism field on a Lie algebra satisfies:
\begin{equation}	\label{eq:4.3}
\begin{aligned}
      (i)&\quad  \Ld_J \AA=\AA\\
     (ii)&\quad  J\jjj \AA=0       \\
    (iii)&\quad  \Im \AA\big|_\OO\cong \Im\ad_x^2  \\
     (iv)&\quad  \Ker \AA\big|_\OO \cong \Ker\ad_x \cap \Im\ad_x
\end{aligned}
\end{equation}
\end{corollary}

Here is a property analogous to the coadjoint representation preserving
the Kirillov-Poisson structure on the dual Lie algebra.

\begin{proposition}	\label{prp: 4.3} 
The endomorphism $\AA$ is preserved by the action of the adjoint
representation, that is
\begin{equation}
      \Ld_{X_v} \AA=0  \qquad\forall v\in L
\end{equation}
\end{proposition}

\begin{proof} 
Use Leibniz rule to show that $(\Ld_{X_v} A)(w)=0$ for
every $w$:
$(\Ld_{X_v} A)(w) = \Ld_{X_v} (A(w))- A \Ld_{X_v} w=
\Ld_{X_v} X_w- A \ldb v, w \rdb =
X_{\ldb v, w\rdb} - X_{\ldb v, w \rdb} = 0$.
\end{proof}

\begin{proposition}	\label{prp: 4.4} 
The endomorphism field on a Lie algebra satisfies:
\begin{equation}	\label{eq:4.5}
\begin{aligned}
      (i)&\quad  \Tr(\AA\om \AA)=K(J,J) \\
       (ii)&\quad  \Tr(\AA)=\chi(J)       \\
      (iii)&\quad  K(\AA v,w)=-K(v,\AA w) 
\end{aligned}
\end{equation}
where the objects are as follows: $K$ is the Killing form defined for two
vectors as $K(v,w)=\Tr\ad_v\om\ad_w$.  When evaluated for $(J,J)$, it
becomes a quadratic scalar function $K(J,J)=x^ax^b c^k_{ai} c^i_{bk}$.
Similarly, $\chi\in L^*$ is a {\it characteristic form} on $L$ defined
$\chi(v)=\Tr\ad_v$.    
Property (iii) states that
the endomorphism $\AA$ is skew-symmetric with respect to the Killing
(possibly degenerated) scalar product.
\end{proposition}

The endomorphism defines for every $k=1,2,\ldots,$ a scalar function of
the power trace
\begin{equation}	\label{eq:4.6}
                  I_k=\Tr \AA^k  = \hbox{Tr}\,(\hbox{ad}_x\om\ldots\om\hbox{ad}_x)
\end{equation}
that will be called {\it Casimir polynomials} on $L$.
In the diagrammatical language they are:

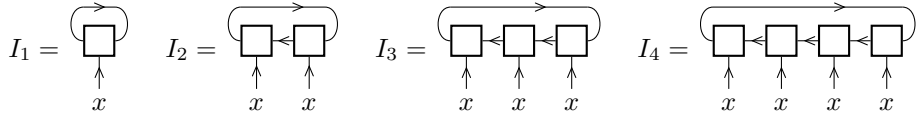
\begin{center}
$I_1 = $
\begin{tikzpicture}[baseline=-14pt,scale=.7, yscale=.8,
decoration={markings, mark=at position .7 with {\arrow{angle 45}};}]    
  \node (c1) at (1,-.5) [rectangle, draw, thick, minimum size=.4cm] {\ ~}; 
  \node (J1) at (1, -2)  {$x$};
  \draw[postaction={decorate}] (J1) to (c1);
  \coordinate (zap1) at (1+.3,.3);  
  \coordinate (zap2) at (1-.3,.3);  
 \draw[-]                                      (c1)  to  [out=0,in=0] (zap1); 
 \draw[postaction={decorate}]  (zap2)   to   (zap1); 
 \draw[-]                                      (zap2)   to  [out=180,in=180] (c1); 
\end{tikzpicture}
\quad $I_2 = $
\begin{tikzpicture}[baseline=-14pt,scale=.7, yscale=.8,
decoration={markings, mark=at position .75 with {\arrow{angle 45}};}]    
  \node (c1) at (1,-.5) [rectangle, draw, thick, minimum size=.4cm] {\ ~}; 
  \node (c2) at (2,-.5) [rectangle, draw, thick, minimum size=.4cm] {\ ~}; 
  \node (J1) at (1, -2)  {$x$};
  \node (J2) at (2, -2) {$x$}; 
  \draw[postaction={decorate}] (J1) to (c1);
  \draw[postaction={decorate}] (J2) to  (c2) ;
  \draw[postaction={decorate}] (c2) to (c1);
  \coordinate (zap1) at (2+.3,.3);  
  \coordinate (zap2) at (1-.3,.3);  
 \draw[-]                                      (c2)  to  [out=0,in=0] (zap1); 
 \draw[postaction={decorate}]  (zap2)   to   (zap1); 
 \draw[-]                                      (zap2)   to  [out=180,in=180] (c1); 
\end{tikzpicture}
\quad$I_3 = $
\begin{tikzpicture}[baseline=-14pt,scale=.7, yscale=.8,
decoration={markings, mark=at position .7 with {\arrow{angle 45}};}]    
  \node (c1) at (1,-.5) [rectangle, draw, thick, minimum size=.4cm] {\ ~}; 
  \node (c2) at (2,-.5) [rectangle, draw, thick, minimum size=.4cm] {\ ~}; 
  \node (c3) at (3,-.5) [rectangle, draw, thick, minimum size=.4cm] {\ ~}; 
  \node (J1) at (1, -2)  {$x$};
  \node (J2) at (2, -2) {$x$}; 
  \node (J3) at (3, -2) {$x$ }; 
  \draw[postaction={decorate}] (J1) to (c1);
  \draw[postaction={decorate}] (J2) to  (c2) ;
 \draw[postaction={decorate}] (J3) to  (c3) ;
  \draw[postaction={decorate}] (c2) to (c1);
  \draw[postaction={decorate}] (c3) to (c2);
  \coordinate (zap1) at (3+.3,.3);  
  \coordinate (zap2) at (1-.3,.3);  
 \draw[-]                                      (c3)  to  [out=0,in=0] (zap1); 
 \draw[postaction={decorate}]  (zap2)   to   (zap1); 
 \draw[-]                                      (zap2)   to  [out=180,in=180] (c1); 
\end{tikzpicture}
\quad $I_4 = $
\begin{tikzpicture}[baseline=-14pt,scale=.7, yscale=.8,
decoration={markings, mark=at position .7 with {\arrow{angle 45}};}]    
  \node (c1) at (1,-.5) [rectangle, draw, thick, minimum size=.4cm] {\ ~}; 
  \node (c2) at (2,-.5) [rectangle, draw, thick, minimum size=.4cm] {\ ~}; 
  \node (c3) at (3,-.5) [rectangle, draw, thick, minimum size=.4cm] {\ ~}; 
  \node (c4) at (4,-.5) [rectangle, draw, thick, minimum size=.4cm] {\ ~}; 
  \node (J1) at (1, -2)  {$x$ };
  \node (J2) at (2, -2) {$x$}; 
  \node (J3) at (3, -2) {$x$ }; 
  \node (J4) at (4, -2) {$x$ }; 
  \draw[postaction={decorate}] (J1) to (c1);
  \draw[postaction={decorate}] (J2) to  (c2) ;
 \draw[postaction={decorate}] (J3) to  (c3) ;
 \draw[postaction={decorate}] (J4) to  (c4) ;
  \draw[postaction={decorate}] (c2) to (c1);
  \draw[postaction={decorate}] (c3) to (c2);
  \draw[postaction={decorate}] (c4) to (c3);
  \coordinate (zap1) at (4+.3,.3);  
  \coordinate (zap2) at (1-.3,.3);  
 \draw[-]                                      (c4)  to  [out=0,in=0] (zap1); 
 \draw[postaction={decorate}]  (zap2)   to   (zap1); 
 \draw[-]                                      (zap2)   to  [out=180,in=180] (c1); 
\end{tikzpicture}
\end{center}
~\\
etc. 
Clearly, the second invariant is a quadratic function related to Killing form and will be denoted 
$\kappa = I_2 = K(J,J) = \kappa$,
but the third is obviously not related to the Lie 3-form.

\begin{corollary}	\label{cor:4.4}  
Differentials of the trace functions are among the annihilators of $\AA$, i.e.,
\begin{equation}	\label{eq:4.7}
                  \AA\jjj dI_k = 0      
\end{equation}
\end{corollary}

\newpage
\section{The endomorphism field and dynamical systems}  

Since the dual Lie algebra $L^*$ with its Poisson structure has deep connections 
with classical mechanics, namely with Hamiltonian formalism, 
one may expect that so does a Lie algebra with its 
endomorphism field $\AA$.
The candidate coming to mind first is Lagrangian mechanics,
as suggested by this chain of correspondences:
\\

\begin{tabular}{ccccc}
KKS theorem       & $\to$   &symplectic & $\to$ &Hamilton     \\
(Lie coalgebras)  &             &geometry  &          &equations   \\[8pt]
Theorem 2.1      & $\to$    &endomorphic &$\to$ &?     \\
(Lie algebras)     &            &geometry    &            
\end{tabular}
\\

\noindent
Duality between tangent bundle $TQ$ over a manifold $M$, 
which possesses enough structure so that any
(``regular'') function $\mathcal L$ on $TQ$ defines a dynamical system via Lagrange equations,
and the cotangent bundle $T^*M$ with its own symplectic structure
$\omega$ granting a Hamiltonian formalism induced by the Hamiltonian $\mathcal H$,
suggests that the question mark in the above diagram of analogies should be replaced by 
some sort of Lagrange formalism.
This guess may be supported by the fact 
that the Lagrange formalism is actually based on the natural endomorphism field 
on the tangent fiber bundle (see Appendix B).

Yet it seems that the most direct 
formalism at the question mark seems  --- much generalized -- Lax equations of motion.

Although Lax equations are typically defined as matrix equations,
the endomorphism $\AA$ allows one to geometrize it in a new way.
In the next sections we shall discuss ``Lax vector fields" on a Lie algebra and
will push the analogy with symplectic geometry to see how far it goes.

We show that, quite pleasantly,  ``Lax vector fields''
form a closed subalgebra under vector field commutator. We shall also
define a new ``Poisson bracket'' in the space of vector fields on Lie
algebra, and  prove a homomorphism between Lie algebra of vector fields with 
this bracket with the standard Lie algebra of vector field.  


\section{The algebra of Lax vector fields}  

Let us start with a general construction.
By analogy to symplectic geometry dealing with manifolds equipped with symplectic structure, $\{M,\,\omega\}$,
we may consider a pair $\{M, \AA\}$ where manifold $M$ is equipped with a structure defined by 
a field of endomorphisms -- (1,1)-variant tensor field on $M$.
Exploring further the analogy, we may study dynamical systems described by vector fields that 
are defined by their ``potentials"  --- other vector-fields.
Thus, instead of Hamilton equations, we have a map
\begin{equation}
\label{eq:liemanifold}
\XX M   \to\  \XX M: \qquad B\  \to\  X_B = B\jjj\AA\ \equiv \ \AA B  .
\end{equation}
This contrast with symplectic geometry, where the potentials of dynamical systems are differential  forms,
namely differentials of Hamiltonians.  It would be natural to require 
that the set of all such dynamical systems ,
$\XX_\AA M =\{\AA Y | Y\in\XX M \}$,
be closed under the Lie bracket of vector fields.  This way it would form a subalgebra  
of $\{ \XX M,\,[\, . \, ,\, .\, ] \}$.
The final demand would be to have a well-defined product of vector fields (potentials) such that the map 
(\ref{eq:liemanifold}) is a homomorphism of the corresponding algebras.

One may ask why one would want to replace one vector field by another: one gain may be that in 
the new form some integrals of motion may be found more easily.

In this section we show that a Lie algebra with the endomorphism field defined
in the previous sections forms such a system. In particular, it is equipped with a bracket for potentials 
that we define below.  
\\ 

Consider the underlying linear space $L$ of a Lie algebra
$\{L,\ldb\;,\;\rdb\}$ as a manifold.  Any smooth
vector field $B$ can be viewed as a generator (or ``{\bf potential}") of a
dynamical system defined by vector field $X_B$ defined
\begin{equation}	\label{eq:6.1}
                     X_B= \AA B
\end{equation}
The integral curves of $X_B$ satisfy the Lax equations, which in 
a somewhat imprecise way are expressed
\[
                    \dot x(t) = [x,B_x]
\]
where the $x$ on the left side is understood as a point in $L$, 
while the $x$ inside the bracket on the right side as a vector in $L$.
More accurately, 
\begin{align*}
              \dot c(t) & = [J_{c(t)}, B_{c(t)}] \\
                        & = \AA\om B\om c(t) 
\end{align*}

\begin{definition}	\label{defn:6.1}
Vector fields on a Lie algebra $L$ of form \eqref{eq:6.1} 
will be called {\bf Lax vector fields} generated by $B$, or Lax dynamical systems.
In the diagrammatic representation, the Lax vector field is:
\begin{equation}
X_B = B\jjj \mathcal{A} = \hspace{0in}
\begin{tikzpicture}[baseline=-4pt,scale=.6,, yscale=.9,
decoration={markings, mark=at position .7 with {\arrow{angle 45}};}]    
  \node (t) at (0,0) [rectangle, draw, thick, minimum size=.5cm] { ~ };
  \node (J) at (-.7,-2) {$x$};
  \node (B) at (.5,-2) [rectangle, draw, thick, minimum size=.5cm] {B };
  \coordinate (u) at (0,1.3);
  \draw[postaction={decorate}] (J) to [out=90,in=-90] ($(t.south) - (.6em,0)$);
  \draw[postaction={decorate}] (B) to [out=90,in=-90] ($(t.south) + (.6em,0)$);
  \draw[-triangle 60] ($(t.north) + (0,0)$) to [out=90,in=-90] (u);
\end{tikzpicture}
\end{equation}
The space of Lax vector fields will be denoted by
$\XX_\AA L =\AA(\XX L) \subset \XX L$.
\end{definition}

A simple and a well-known fact is the existence of Casimir invariants:

\begin{corollary}	\label{cor:6.2}
The dynamical system defined by a Lax vector field \eqref{eq:6.1} leaves
Casimir polynomials $I_k$ invariant, $X_B I_k=0$, for any $B\in \XX L$.
\end{corollary}

\noindent
{\bf Proof:} (graphical)  We show the reasoning for $I_2=K(J,J)$
(quadratic polynomials defined by Killing form): 

\begin{center}
$X_B I_2 = (X_B\otimes J) \jjj I_2 = 2 \times \ $
\begin{tikzpicture}[baseline=-14pt,scale=.7,
decoration={markings, mark=at position .7 with {\arrow{angle 45}};}]    
  \node (JJ) at (1.8, -1.5) { }; 
  \node (JJJ) at (0, -2) {$x$};
  \node (c1) at (.8,-.5) [rectangle, draw, thick, minimum size=.4cm] {\ ~}; 
  \node (c2) at (1.8,-.5) [rectangle, draw, thick, minimum size=.4cm] {\ ~}; 
  \draw[postaction={decorate}] (JJJ) to [out=90,in=-90] (c1);
\draw[-o] (c2) to (JJ) ;
  \draw[postaction={decorate}] (c1) to (c2);
  \coordinate (zap1) at (1.8+.3,.3);  
  \coordinate (zap2) at (.8-.3,.3);  
 \draw[-]                                      (c2)  to  [out=0,in=0] (zap1); 
 \draw[postaction={decorate}]  (zap1)   to   (zap2); 
 \draw[-]                                      (zap2)   to  [out=180,in=180] (c1); 
  \node (qc) at (0+1.8, 0-2.2) [rectangle, draw, thick, minimum size=.4cm] { ~ };
  \node (qJ) at (-.6+1.8, -2-1.8) {$x$ };
  \node (B) at (.7+1.8, -2-1.8) [rectangle, draw,  minimum size=.4cm] { B };
  \coordinate (qu) at (0+1.8,1.3-2.2-.4);
  \draw[-triangle 60] (qc) to (qu);
  \draw[postaction={decorate}] (qJ) to  [out=90,in=-90]  ($(qc.south) - (.4em,0)$);
  \draw[postaction={decorate}] (B) to [out=90,in=-90] ($(qc.south) + (.4em,0)$);
\end{tikzpicture}
$ =\  \frac{2}{3}\ $
\begin{tikzpicture}[baseline=-14pt,scale=.7,
decoration={markings, mark=at position .7 with {\arrow{angle 45}};}]    
  \node (B) at (3.5, -2) [rectangle, draw,  minimum size=.5cm] { B };
  \node (JJ) at (2, -2) {$x$}; 
       \node(alt) at (2,.1) {alt};
  \node (JJJ) at (1, -2)  {$x$};
  \node (c1) at (1,-.5) [rectangle, draw, thick, minimum size=.4cm] {\ ~}; 
  \node (c2) at (2,-.5) [rectangle, draw, thick, minimum size=.4cm] {\ ~}; 
  \node (c3) at (3,-.5) [rectangle, draw, thick, minimum size=.4cm] {\ ~}; 
  \draw[postaction={decorate}] (JJJ) to (c1);
  \draw[postaction={decorate}] (B) to [out=90,in=-90] (c3);
\draw[postaction={decorate}] (JJ) to  (c2) ;
  \draw[postaction={decorate}] (c1) to (c2);
  \draw[postaction={decorate}] (c2) to (c3);
  \coordinate (zap1) at (3+.3,.3);  
  \coordinate (zap2) at (1-.3,.3);  
 \draw[-]                                      (c3)  to  [out=0,in=0] (zap1); 
 \draw[postaction={decorate}]  (zap1)   to   (zap2); 
 \draw[-]                                      (zap2)   to  [out=180,in=180] (c1); 
\end{tikzpicture}
$\ = \ 0$.
\end{center}

\noindent
where first we used Jacobi identity (\ref{eq:jacobi}) and then skewsymmetry of the resulting $\omega$. The right side 
vanishes as $\omega$ has two identical entries, $x$. 
The argument for the other Casimir invariants is similar. 
$\square$

The geometric meaning of the fundamental Nijenhuis property of 
the endomorphism field becomes clear in the current context.  
Namely, it implies that the space of Lax vector fields $\XX_\AA L$ is closed 
under the commutator of vector fields $[ \XX_\AA L, \XX_\AA L] \subset \XX_\AA L$.  
A new bracket of vector fields is implied.


\begin{theorem}		\label{th:6.3}
The space of Lax vector field  forms a subalgebra of the algebra 
of smooth vector fields, $\XX_\AA L <\XX L$.
In particular, if $X_B$ and $X_C$ are two (global) Lax vector fields,
then their commutator is a Lax vector field with potential
\begin{equation}	\label{eq:6.3}
       \{B,C\} =:  -[\![B,C]\!]+[X_B,C] +[B,X_C]- X_{[B,C]}
\end{equation}
so that there is an homomorphism between the Lax vector fields with the
regular vector field commutator and  all vector fields with  $\{\;,\;\}$ product:
\begin{equation}	\label{eq:6.4}
  [X_B,X_C] = X_{\{B,C\} }
\end{equation}
\end{theorem}

\begin{proof}
This follows from the fact that $[\AA,\AA]$ is proportional
to $\AA$.  Rewrite the definition of the Nijenhuis bracket \eqref{eq:3.7} for
$\AA$ and use Theorem \ref{th:2.1}:
\begin{equation}
\begin{array}{rl}
    [X_B,X_C]
          & = [\AA B,\AA C] \\
          & = {1\over 2}[\AA,\AA](B,C) +\AA[\AA B,C] +\AA[B,\AA C]-\AA^2[B,C]\\
          & = -\AA[\![B,C]\!] +\AA [X_B,C] +\AA [B,X_C]-\AA X_{[B,C]}\\
          & = \AA\;\left(\;
                     - [\![B,C]\!]+[X_B,C] +[B,X_C]- X_{[B,C]}
                       \; \right)
\end{array}
\end{equation}
where in the last part we see that the endomorphism field $\AA$ may be "factored out" 
thanks to Theorem \ref{th:2.1}.  Thus the commutator is of the form \eqref{eq:6.1}
the formulas in the theorem follow.
\end{proof}

\begin{proposition}	\label{prp:6.5}
The bracket \eqref{eq:6.3} can be calculated by the
following formula
\begin{equation}	\label{eq:6.5}
\{A,B\} =  \ldb A,\,B\rdb + \underbrace{X_AB-X_BA}_{(A,B)}
\end{equation}
where $X_AB=x^iA^j\;c_{ij}^k\; \q_kB^p\;\q_p$.

\end{proposition}

Notice that although the two right-most terms are defined in coordinates,
their difference has a coordinate-free meaning, as it can be defined
by $X_AB-X_BA = \{A,B\} - \ldb A,\,B\rdb$.
\\

The bracket $\{\;,\;\}$ turns the space of vector fields on $L$ into a
Lie algebra and can be viewed as a ``differential deformation"
of the Lie algebra bracket $\ldb\;,\;\rdb$.
Due to its involved nature, it may be a rather surprising that
defines a Lie algebra.  The Jacobi identity is not a direct consequence
and results by intertwined interaction of the Jacobi identities of the Lie algebra $L$ 
and of the Lie algebra of vector fields.

\begin{remark}	\label{rem:6.4}
For two constant vector fields $\wtilde v$ and $\wtilde w$ 
that extend vectors $v,w\in L$, it is 
$\{\wtilde v,\wtilde w\}=\wtilde{\;\ldb v, w\rdb\;}$.
Thus the bracket formula \eqref{eq:6.4} reduces in this case to the infinitesimal
representation $[X_v,X_w]=X_{\ldb v,w \rdb}$.
\end{remark}

\begin{theorem}		\label{th:6.7}
The pair $\{\XX L,\; \{\;\cdot\;,\;\cdot\;\} \;\}$
forms a Lie algebra, i.e., the bracket (\ref{eq:6.3},\ref{eq:6.6}) of vector fields satisfies
the following properties:
\begin{equation}	\label{eq:6.6}
\begin{aligned}
            (i) &\quad \hbox{(linearity)} \\
            (ii) &\quad \{A,B\}=-\{B,A\} \qquad\hbox{(skewsymmetry)} \\
           (iii) &\quad  \{A,\{B,C\}\}+\{B,\{C,A\}\}+\{C,\{A,B\}\} = 0
                      \quad\hbox{(Jacobi identity)} 
\end{aligned}
\end{equation}
\end{theorem}

\begin{proof}
If $X,Y\in\XX L$ are two vector fields, then we denote
$X\trr Y= X^i (\q_i Y^j) \q_j$ a vector field calculated in linear
coordinate system.  Thus, formula \eqref{eq:6.5} can be written as
\[
\{A,B\}= \ldb A,\, B\rdb +(A,\, B)
\]
where
\[
(A,\,B) = X_A\trr B - X_B\trr A
\]
Now, using the formula $X_A=\ldb x,\, A\rdb$, we get
\begin{equation}	\label{eq:*}
\begin{aligned}
\{\{A,\, B\},\,C\}
&= \{ \ldb A,\, B\rdb +(A,\,B), \, C\} \\
&= \underbrace{\ldb \ldb A,\, B \rdb, \, C\rdb}_{(a)}
 + \underbrace{\ldb ( A,\, B ), \, C\rdb}_{(b)}
 + \underbrace{( \ldb A,\, B \rdb, \, C )}_{(c)}
 + \underbrace{(( A,\, B ), \, C )}_{(d)} \cr
&= \underbrace{\ldb \ldb A,\, B \rdb, \, C\rdb}_{(0)}       \qquad(a) \\
& + \underbrace{ \ldb\ldb  x,\, A \rdb\trr B,\, C\rdb}_{(1)}
  - \underbrace{ \ldb\ldb  x,\, B \rdb\trr A,\, C\rdb}_{(2)} \qquad(b) \\
& + \underbrace{ \ldb x ,\, \ldb  A,\, B \rdb \rdb \trr C }_{(5)}
  - \underbrace{ \ldb\ldb  x,\, C \rdb\trr A,\, B\rdb}_{(1)}
  - \underbrace{ \ldb A,\, \ldb  x,\, C \rdb\trr B \rdb}_{(2)} \qquad(c) \\
& + \underbrace{ \ldb x,\, \ldb  x,\, A \rdb\trr B \rdb\trr C}_{(3)}
  - \underbrace{ \ldb\ldb  x,\, C \rdb,\, A\rdb \trr B }_{(5)}
  - \underbrace{ \ldb x,\, \ldb  x,\, C \rdb\trr A \rdb \trr B }_{(3)} \qquad(d) \\
& - \underbrace{ \ldb x,\, \ldb  x,\, B \rdb\trr A \rdb \trr C }_{(4)}
  + \underbrace{ \ldb\ldb  x,\, C \rdb,\, B \rdb \trr A }_{(5)}
  + \underbrace{ \ldb x,\, \ldb  x,\, C \rdb \trr B \rdb \trr A }_{(4)} \qquad(d)
\end{aligned}			\tag{$*$}
\end{equation}
where the letters (a), (b), (c), and (d) are used to indicate the origin
of terms in the second part of the equation.
The sum
\[
\{\{A,\, B\},\,C\} +
\{\{B,\, C\},\,A\} +
\{\{C,\, A\},\,B\}
\]
contains every term of equation \eqref{eq:*} in each of the three cyclic
permutations of A,B, and C. The sum of such terms marked by any of the
numbers (1) to (4) vanish due to opposite signs.  The group of terms
marked by (0) and terms marked by (5) both vanish, each due to the Jacobi
identity of the Lie algebra product.
\end{proof}

\begin{corollary}
\label{cor:6.88}
There is a Lie algebra homomorphism 
$\{ \XX L, \{ \,.\,,\,.\, \} \} \to \{\XX_\AA,\, [\,.\,,\,.\,] \}$
between Lie algebra of vector fields on $L$ with bracket defined by (\ref{eq:6.3})
and the Lie algebra vector fields on $L$ restricted to Lax vector fields.
\end{corollary}

By analogy to Hamiltonian formalism of classical mechanics
we have a property that may be viewed as a counterpart of Poisson Theorem: 

\begin{corollary}[``\`a la Poisson'']	\label{cor:6.8}
If vector fields $B$ and $C$ are Lax
potentials of symmetries of a dynamical system, then so is $\{B,C\}$.
\end{corollary}

\begin{proof} Use the Jacobi identity for vector fields
\[
 [X_B,\; \underbrace{[X_C,X]}_{0} \;] +
 [X,\; \underbrace{[X_B,X_C]}_{X_{ \{B,C\} } }\; ] +
 [X_C,\; \underbrace{[X,X_B]}_{0} \;] = 0
\]
hence the claim: $\Ld_{X_{\{B,C\}}} X = 0$.
\end{proof}

\noindent
{\bf Basic examples.} What can be used as a Lax potential?
The simplest are constant vector fields, 
in which case the homomorphism reduces to Proposition \ref{prop:homo}
(see Remark \ref{rem:6.4}).  
Also, a Lax vector field may be ``reused" as a potential
for a new Lax vector field.  
The following formulas for bracket $\{\cdot,\cdot\}$ may be useful 
for such dynamical systems
\[
\{ v,w \} = \ldb v,w \rdb \qquad
\{ X_v, w \} = X_{\ldb v,w \rdb} \qquad
\{ X_v,X_w \} = X_{\AA \ldb v,w \rdb } - \ldb X_v,X_w \rdb   
\]
where $v$ and $w$ are understood as constant vector fields (the tilde is suppressed for simplicity).

Another class consists of  Lax vector fields  
generated from linear vector fields on $L$.
Euler's equations of the motion a rigid body belongs to this category.
Here is their --- somewhat na\"ive --- generalization to arbitrary Lie algebra:
Let $R\in\hbox{End}\, L$ be a matrix describing  the tensor of inertia.
If vector field $J\jjj \wtilde R$ is used as a "potential",
the resulting Lax vector field $X = \AA(J\jjj \wtilde R)$ describes the dynamical system of ``rotating body".
In the case of the Lie algebra of 3-dimensional orthogonal group
$L=so(3,\mathbb R)$, with the standard coordinates $(x,y,z)$, 
and for a diagonal matrix $R = \hbox{diag}(a,b,c)$, we get 
the standard Euler's equations
\[
X = (b-a)xy\q_z + (c-b)yz\q_x + (a-c)zx\q_y 
\] 
(or $\dot x = (c-b)zy$, etc.).
A more accurate description will be given elsewhere.

\newpage
\section{Analogies and dualities}   
\label{s:duality}
 
The analogies between differential geometry (calculus) on a Lie algebra and on a Lie coalgebra
are shown in the following table.  Note that the  Lie algebra structure $c$ is a (1,2)-variant tensor on 
the Lie algebra $L$, but it is a (2,1)-variant tensor on the dual space $L^*$.
This results in quite different calculus on both spaces treated as manifolds.

\begin{table}[h!]
\begin{tabular}{lll}
\hline\\[-7pt]
                                                             &\bf Lie algebra $L$              &\bf   Lie coalgebra $L^*$ \\
                                                             &\bf as manifold                    &\bf   as a manifold\\[3pt]
\hline\\[-6pt]
Coordinates                                          &$\{x_i\}$                     &$\{ x^i\}$                         \\[4pt]
Constant structure tensor                     & $\lambda=\wtilde c$       &$\lambda' = \wtilde c$\\
\qquad \small  ...in coordinates              &$\phantom{\lambda}={\scriptstyle\frac{1}{2}}
                                                                                                      c^k_{ij}\, \q_k\otimes dx^i\w dx^j$      \qquad   
                                                                                                     &$\phantom{\lambda'}={\scriptstyle\frac{1}{2}}
                                                                                                                  c^k_{ij}\, dx_k\otimes \q^i\w \q^j$\\[7pt]
Jacobi vector field                                &$J= x_i\q^i$          &$J' = x^i\q_i$                    \\[7pt]
Primary differential object                    &$\AA = J\jjj  \lambda$        &  $\Omega = J ' \jjj \lambda'$\\
\qquad \small ...in coordinates                                  &$\phantom{\AA}={\scriptstyle\frac{1}{2}}
                                                                                                              c^k_{ij} x^i\, \q_k\otimes dx^j$      \qquad   
                                                                                   &$\phantom{\Omega}={\scriptstyle\frac{1}{2}}
                                                                                                             c^k_{ij} x_k\,  \q^i\w \q^j$\\
                                                             &\small(endomorphism field)      &\small (Poisson structure) \\[7pt]  
Basic rule                                              &$[\AA,\AA] =  -2\AA\jjj\lambda$    &$[\Omega,\Omega] =  0$\\[-3pt]
\ \small(consequence of Jacobi identity) \\[7pt]

``A potential''                                          &$B\in \XX L$ (vector field)  &$H\in \FF L$ (function) \\
\quad ...generates dynamical system   \qquad                     &$X_B = B\jjj \AA \equiv \AA B $  &$X_H = dH\jjj \Omega$\\ 
                                                                           &\small(Lax equations)     &\small(Hamilton equations) \\[5pt]    
\hline\\[-18pt]
\end{tabular}
\caption{Legend:  $\q_i \equiv \q/\q x_i$, $\q^i \equiv \q/\q x^i$.  
Tilde $\sim$ denotes extension of tensors to tensor fields
on $L$ and on $L^*$, 
defined by the affine structure on linear spaces.}
\end{table}

On $L^*$ as a manifold, $\lambda$ is (2.1) variant.
In pictures, the canonical Poisson structure on $L^*$ and Hamiltonian mechanics may be illustrated as follows:
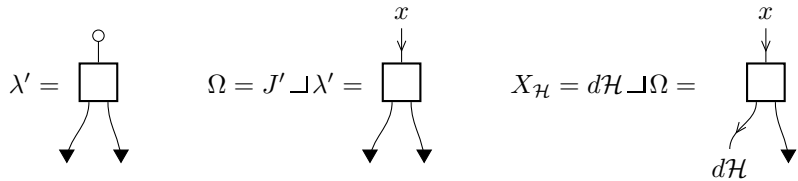
\begin{center}
$\lambda' =$ 
\begin{tikzpicture}[baseline=-4pt,scale=.6,, yscale=.9,
decoration={markings, mark=at position .7 with {\arrow{angle 45}};}]    
  \node (t) at (0,0) [rectangle, draw, thick, minimum size=.5cm] { ~ };
  \coordinate (J) at (-.7,-2) ;
  \coordinate (JJ) at (.5,-2);
  \coordinate (u) at (0,1.3);
  \draw [-triangle 60] ($(t.south) - (.6em,0)$) to  [out=-90,in=90](J);
  \draw [-triangle 60] ($(t.south) + (.6em,0)$) to  [out=-90,in=90](JJ);
  \draw[-o] (t) to [out=90,in=-90] (u);
\end{tikzpicture}
\qquad\quad
$\Omega = J'\jjj\lambda' = $
\begin{tikzpicture}[baseline=-4pt,scale=.6,, yscale=.9,
decoration={markings, mark=at position .7 with {\arrow{angle 45}};}]    
  \node (t) at (0,0) [rectangle, draw, thick, minimum size=.5cm] { ~ };
  \coordinate (J) at (-.7,-2);
  \coordinate (JJ) at (.5,-2);
  \node (u) at (0,1.7) {$x$};
  \draw [-triangle 60] ($(t.south) - (.6em,0)$) to  [out=-90,in=90](J);
  \draw [-triangle 60] ($(t.south) + (.6em,0)$) to  [out=-90,in=90](JJ);
  \draw[postaction={decorate}]  (u) to [out=-90,in=90] (t);
\end{tikzpicture}
\qquad\quad
$X_\mathcal H = d\mathcal H \jjj \Omega =$ 
\begin{tikzpicture}[baseline=-4pt,scale=.6,, yscale=.9,
decoration={markings, mark=at position .7 with {\arrow{angle 45}};}]    
  \node (t) at (0,0) [rectangle, draw, thick, minimum size=.5cm] { ~ };
  \node (J) at (-.8,-2.1)  {$d\mathcal H$};
  \coordinate (JJ) at (.5,-2);
  \node (u) at (0,1.7) {$x$};
  \draw [postaction={decorate}]  ($(t.south) - (.6em,0)$) to  [out=-90,in=90](J);
  \draw [-triangle 60] ($(t.south) + (.6em,0)$) to  [out=-90,in=90](JJ);
  \draw[postaction={decorate}]  (u) to [out=-90,in=90] (t);
\end{tikzpicture}

\end{center}

\section{Remark on Lagrange equations}   

%

While the cotangent bundle $T^*Q$ over a manifold $Q$
possesses a canonical differential biform $\omega\in\La^2Q$ defining symplectic
structure, the tangent bundle $TQ$ possesses a canonical  (1,1)-variant tensor
field $S\in\TT^{(1,1)}TQ$ defining an endomorphism field (endomorphisms of $T(TQ)$ and $T^*(TQ)$).
In the natural coordinates $(x^i,v^i)$ on $TQ$, this tensor can be
expressed as $S={\q\over \q v^i} \otimes dx^i$ (sum over $i$).
Its basic property is $\Ker S = \Im S$ (implying nilpotence $S\om S=0$).
If $\mathcal L$ is a function on $TQ$ (a Lagrangian), 
then one defines a biform $\omega=d\om S\om d\,\mathcal L$, which for a ``regular'' Lagrangian is
nondegenerate and therefore forms a symplectic structure.
It is easy to see that Lagrange equations
may be written as 
\[
\Ld_X (S\om d\mathcal L) = d\mathcal L
\]
The existence of $S$ and its role in Lagrangian mechanics was noticed rather late \cite{MS};
it replaces a rather awkward notion of ``vertical derivative" 
used before in an attempt to geometrize Euler-Lagrange equations \cite{AM}. 

In a series of papers \cite{Cr1, Cr2},
a notion of {\it almost tangent structure} on a differential manifold $M$
has been introduced, as a tensor $S\in\TT^1_1M$ that satisfies
\begin{equation}	\label{eq:1.2}
\begin{aligned}
     (i)&\quad \Ker S = \Im S  \qquad (\Rightarrow\quad S\om S=0)\\
    (ii)&\quad [S,S]  = 0 
\end{aligned}
\end{equation}
where the second condition (ii) is a generalization of
the Schouten-Nijenhuis bracket to ``vector-valued differential
forms'' (see e.g. \cite{Tu2} and \cite{YI}),
which assures (local) integrability of the distribution  $\Ker S$.
As a result, one obtains all of the structure of the tangent bundle
($\Ker S$ gives the fibering) except distinguishing
the zero-section.  

A Lie algebra may provide an example of a generalized version of 
such an Euler-Lagrange structure, in which the above conditions (\ref{eq:1.2})
are relaxed.
\[
\begin{matrix}
\hbox{Hamilton} &\to &\hbox{cotangent bundle} &\to
                     &\hbox{symplectic} &\to &\hbox{KKS theorem}     \\
\hbox{equations}&    & T^*Q                   &
                     &\hbox{geometry}  &    &\hbox{(Lie coalgebras)}\\[8pt]
\hbox{Lagrange}      &\to &\hbox{tangent bundle}   &\to
                     &\hbox{endomorphic} &\to &\hbox{Theorem 2.1}     \\
\hbox{equations}&    & TQ                     &
                     &\hbox{geometry}  &    &\hbox{(Lie algebras)}  
\end{matrix}
\]

Whether such potential relationship between Lie algebras and generalized Lagrangian formalism
would be fruitful is an interesting question in the context of geometric quantization
and representation theory known for coadjoint orbits in the Lie co-algebras.


\section*{Acknowledgements}

The author would like to thank Zbigniew Oziewicz for discussions that much improved the presentation 
and Philip Feinsilver for many valuable comments and suggestions.

\newpage





\end{document}